\documentclass[aps,pra,amsmath,amssymb,showpacs,twocolumn]{revtex4}

\usepackage[T1]{fontenc}
\usepackage{amssymb,amsmath}
\usepackage{amsfonts}
\usepackage{graphicx}
\usepackage{amsthm}
\usepackage{color}
\usepackage{mathbbol}

\def\ket#1{|#1\rangle}

\def\bra#1{\langle#1|}

\def\tr{\mathrm{tr}}

\newcommand{\bq}{{\mathbf q}}
\newcommand{\bn}{{\mathbf n}}
\newcommand{\cohalpha}{\alpha}

\newcommand{\bJ}{{\bf J}}
\newcommand{\hbq}{\hat{{\bf q}}}
\newcommand{\dqJ}{2{\bf q}\boldsymbol{\cdot}{\bf J}}

\newcommand{\bx}{{\bf x}}
\newcommand{\bu}{{\bf u}}

\newtheorem{theorem}{Theorem}[]

\newtheorem{lemma}[]{Lemma}

\begin{document}

\title{Tensor Representation of Spin States}

\author{O.~Giraud$^{1}$, D.~Braun$^{2}$, D.~Baguette$^{3}$, T.~Bastin$^{3}$, J.~Martin$^{3}$}
\affiliation{$^{1}$\mbox{LPTMS, CNRS and Universit\'e Paris-Sud, UMR 8626, B\^at.~100,
91405 Orsay, France}\\
$^{2}$Institut f\"ur theoretische Physik, Universit\"at T\"{u}bingen,
72076 T\"ubingen, Germany\\
$^{3}$\mbox{Institut de Physique Nucl\'eaire, Atomique et de Spectroscopie,~Universit\'e de Li\`ege, 4000 Li\`ege, Belgium}
}

\begin{abstract}
We propose a generalization of the Bloch sphere representation for
arbitrary spin states. It provides a
compact and elegant representation of spin  density matrices 
in terms of tensors that share the most important properties of
Bloch vectors. Our representation, based on covariant matrices introduced by Weinberg in the context of quantum field theory, allows for a simple parametrization of  coherent  spin states, and a straightforward transformation of density matrices under 
local unitary and partial tracing operations. It enables us to provide a criterion for anticoherence, relevant in a broader context such as quantum polarization of light.

\end{abstract}
\date{September 3, 2014 - Revised: February 26, 2015}


\maketitle

The concept of spin is ubiquitous in quantum theory and all related fields of research, such as solid-state physics, molecular, atomic, nuclear or high-energy physics~\cite{Perdew86,Yu50,Pauli25,Ohlsen72,Adamczyk14}. It has profound implications for the structure of matter as a consequence of the celebrated spin-statistics theorem~\cite{Pauli40}. The spin of a quantum system, be it an electron, a nucleus or an atom, has also been proven to be a key resource for many applications such as in spintronics~\cite{Zutic04}, quantum information theory~\cite{Nie00} or nuclear magnetic resonance~\cite{Rabi38}. Simple geometrical representations of spin states~\cite{Goyal11} allow one to develop physical insight regarding their general properties and evolution. Particularly well studied is the case of a single two--level system, formally equivalent to a spin-1/2. In this case, the geometric representation is particularly simple. Indeed, the density matrix can be expressed in a basis formed of Pauli matrices and the identity matrix, leading to a parametrization in terms of a vector in ${\mathbb R}^3$. Pure states correspond to points on a unit sphere, the so-called Bloch sphere, and mixed states fill the inside of the sphere, the ``Bloch ball''.  The simplicity of this representation help visualize
the action and geometry of all possible spin-$1/2$ quantum channels~\cite{Zyczkowski_book}. For arbitrary pure spin states, another nice geometrical representation has been developed by Majorana in which a spin-$j$ state is visualized as $2j$ points on the Bloch sphere~\cite{Majorana32}. This so-called Majorana or stellar representation has been exploited in various contexts (see, e.~g.,~\cite{Bloch45,Mar10,Markham11,Giraud10,Bruno12,Zyczkowski_book}), but cannot be generalized to mixed spin states.

Given the importance of geometrical representations, there have been numerous attempts to extend the previous representations to arbitrary mixed states. The former rely on a variety of sophisticated mathematical concepts such as $su(N)$-algebra generators~\cite{SUNpos,SUN,Goyal11}, polarization operator basis~\cite{Aga81,Polapos,Pola}, Weyl operator basis~\cite{Weyl}, quaternions~\cite{Mosseri01}, octonions~\cite{Ber03} or Clifford algebra~\cite{Dietz06}. In the present Letter we propose an elegant generalisation to arbitrary spin-$j$ of the spin-1/2 Bloch sphere representation based on matrices introduced by Weinberg in the context of relativistic quantum field theory~\cite{Weinberg}. The main result of the paper is theorem 2, which allows us to express any spin-$j$ density matrix as a linear combination of matrices with convenient properties.
The remarkable features of our representation are especially reflected in
the simple coordinates of coherent states, transformation under SU(2)
operations,  and the simplicity of the representation of reduced density matrices.
To illustrate the usefulness of such a representation, we show that it allows us to give an easy characterisation of anticoherent spin states. Such states have been studied in various contexts, such as quantum polarization of light (see e.g.~\cite{Hoz13,SanchezSoto13}), spherical designs~\cite{Crann}, as well as in the search for maximally entangled symmetric states~\cite{BBM14}. We believe that our representation should prove useful in many of the contexts where the spin formalism is used.

We construct this parametrization of a spin-$j$ density matrix $\rho$ from the set of $4^{2j}$ covariant matrices~\cite{Weinberg}.   
Defining the 4-vector $q=(q_{0},q_{1},q_{2},q_{3})\equiv (q_0,\bq)$,
Weinberg's covariant matrices $S_{\mu_1\mu_2\ldots\mu_{2j}}$, with
$0\leqslant \mu_i\leqslant 3$, are constructed from products of components of $\mathbf{J}=(J_1,J_2,J_3)$ with $J_a$ ($1\leqslant a\leqslant 3$), the
usual $(2j+1)$-dimensional representations of angular momentum operators. They
can be obtained by expanding the square of the $(2j+1)$-dimensional matrix corresponding to the $(j,0)$ 
representation of a Lorentz boost in direction ${\bf q}$,
which can be put in the form~\cite{Weinberg}
\begin{equation}
\label{lorentzboost}
 \Pi^{(j)}(q) \equiv(q_0^2-|\bq|^2)^j\,e^{-2\eta_q\,\hbq\boldsymbol{\cdot}\bJ}
\end{equation}
with $\eta_q=\mathrm{arctanh}(-|\bq|/q_0)$ and $\hbq=\bq/|\bq|$. Matrices $S_{\mu_1\mu_2\ldots\mu_{2j}}$ 
are defined in~\cite{Weinberg}  by identifying the coefficients of the
multivariate polynomial with variables $q_{0},q_{1},q_{2},q_{3}$ in \eqref{lorentzboost} with those of the polynomial
\begin{equation}
\label{egaliteweinberg}
 \Pi^{(j)}(q)= (-1)^{2j}q_{\mu_1}q_{\mu_2}\ldots q_{\mu_{2j}}S_{\mu_1\mu_2\ldots\mu_{2j}}
\end{equation}
(we use Einstein summation convention for repeated indices). An explicit expression for $\Pi^{(j)}(q)$  is given in~\cite{Weinberg} as
\begin{widetext}
\begin{equation}
\label{pientier}
\Pi^{(j)}(q)=(q_0^2-{\bf q}^2)^j+\sum_{k=1}^{j}\frac{(q_0^2-{\bf q}^2)^{j-k}}{(2k)!}(\dqJ)\left(\prod_{r=1}^{k-1}[(\dqJ)^2-(2r\bq)^2]\right)(\dqJ+ 2k q_0)
\end{equation}
for integer $j$, and
\begin{equation}
\label{pidemi}
\Pi^{(j)}(q)=(q_0^2-{\bf q}^2)^{j-1/2}(-q_0-\dqJ)-\sum_{k=1}^{j-1/2}\frac{(q_0^2-{\bf q}^2)^{j-1/2-k}}{(2k+1)!}\left(\prod_{r=1}^{k}[(\dqJ)^2-((2r-1)\bq)^2]\right)(\dqJ+(2k+1)q_0)
\end{equation}
for half-integer $j$. The identity matrix is implicit in front of constant terms. For instance, identifying the coefficient of $q_0^{2j}$ in these expressions we get that $S_{00\ldots 0}$ is the $(2j+1)$-dimensional identity matrix, $\mathbb{1}_{2j+1}$. 
\end{widetext}
The matrices $S_{\mu_1\mu_2\ldots\mu_{2j}}$ are Hermitian matrices, invariant under permutation of indices, and they obey the following linear relation:
\begin{equation}
\label{traceless}
g_{\mu_1\mu_2}S_{\mu_1\mu_2\ldots\mu_{2j}}=0,
\end{equation}
where $g\equiv\mathrm{diag}(-,+,+,+)$.

Let us briefly consider the simplest examples. From \eqref{pidemi}, the explicit expression of $\Pi^{(j)}(q)$ for spin-1/2 reads
\begin{equation}
\label{pi12}
\Pi^{(1/2)}(q)=-q_0-\dqJ
\end{equation}
where $J_a$ are spin-1/2 representations of the angular momentum operators. Identifying with \eqref{egaliteweinberg} directly gives $S_{0}=\sigma_0$ and $S_{a}=2J_a=\sigma_a$ where $\sigma_0$ is the $2\times 2$ identity matrix and $\sigma_a$ are the usual Pauli matrices. The usual Bloch sphere representation for an arbitrary spin-1/2 density matrix $\rho=\frac12\sigma_0+\frac12{\bf x}\boldsymbol{\cdot}\boldsymbol{\sigma}$ can then be expressed in terms of the $S_{\mu_1}$ ($0\leqslant \mu_1\leqslant 3$) as
\begin{equation}
\rho=\frac{1}{2}\,x_{\mu_1} S_{\mu_1}
\label{canonrhoj12}
\end{equation}
with the Bloch vector ${\bf x}=\tr(\rho\,{\boldsymbol{\sigma}})$ and $x_0=1$.

For $j=1$, the equality between expressions \eqref{egaliteweinberg} and \eqref{pientier} for $\Pi^{(1)}(q)$ reads
\begin{equation}
(q_0^2-{\bf q}^2)+2{\bf q}\boldsymbol{\cdot}{\bf J}\left({\bf q}\boldsymbol{\cdot}{\bf J}+q_0\right)=q_{\mu_1}q_{\mu_2}S_{\mu_1\mu_2}.
\end{equation}
Identifying coefficients of this quadratic form yields
$S_{00}=J_0$, $S_{a0}=J_a$ and $S_{ab}=J_{a}J_{b}+J_{b}J_{a}-\delta_{ab}J_0$ with $J_0$ the $3\times 3$ identity matrix. Again, the set of $S_{\mu_1\mu_2}$ matrices can serve to express any spin-1 density matrix $\rho$ as
\begin{equation}
\rho=\frac14\,x_{\mu_1\mu_2}S_{\mu_1\mu_2}
\end{equation}
with coordinates 
\begin{equation}
\label{t1}
x_{\mu_1\mu_2}=\tr(\rho\,S_{\mu_1\mu_2}).
\end{equation}

Expressions \eqref{pientier}--\eqref{pidemi} can be used to generalize this expansion to arbitrary $j$, as we will show in Theorem 2. 
The main property of the covariant matrices is given by Theorem 1 below. We first give a useful lemma.

\begin{lemma}
\label{lemmedaniel}
Let 
$\ket{\cohalpha}$ be a spin-$j$ coherent state, defined for 
$\alpha=e^{-i\varphi}\cot(\theta/2)$ with
$\theta\in[0,\pi]$ and $\varphi\in [0,2\pi[$ by
\begin{equation}
\label{coh}
\ket{\cohalpha}=\sum_{m=-j}^{j}\sqrt{\binom{2j}{j+m}}\left[\sin\!\tfrac{\theta}{2}\right]^{j-m}\left[\cos\!\tfrac{\theta}{2}e^{-i\varphi}
\right]^{j+m}\ket{j,m}
\end{equation}
in the standard angular momentum basis $\{\ket{j,m}:-j\leqslant m \leqslant j\}$, and let $\bn=(\sin\theta\cos\varphi,\sin\theta\sin\varphi,\cos\theta)$. Then
\begin{equation}
\label{eq:prop}
\bra{\cohalpha}\Pi^{(j)}(q)\ket{\cohalpha}=(-1)^{2j}(q_0+\bq\boldsymbol{\cdot}\bn)^{2j}.
\end{equation}
\end{lemma}
The proof of this lemma is based on the SU(2) disentangling theorem
and can be found in the Supplemental Material. One of its consequences is that, by identifying coefficients of the polynomial in $q_{\mu}$ in \eqref{eq:prop}, we get
\begin{equation}
  \label{eq:Stn}
\langle 
\alpha|S_{\mu_1\mu_2\ldots\mu_{2j}}|\alpha\rangle=n_{\mu_1}n_{\mu_2}\ldots
n_{\mu_{2j}},
\end{equation}  
with $n_0=1$.

In the Majorana representation, any pure spin-$j$ state is viewed as a permutation symmetric state of a system of $N\equiv 2j$ spin-$1/2$, or equivalently as an $N$-qubit symmetric state. The Hilbert space $\mathcal{H}\equiv \mathbb{C}^{2^N}$ of an
$N$ spin-$1/2$ system has dimension $2^N$ but its symmetric subspace
$\mathcal{H}_S$ has only dimension $N+1=2j+1$. It is spanned by
the symmetric Dicke states
\begin{equation}
\ket{D_{N}^{(k)}}=\mathcal{N}\sum_{\pi}\ket{\underbrace{\downarrow\ldots \downarrow}_{N-k}\underbrace{\uparrow\ldots \uparrow}_{k}},
\quad k = 0, \ldots, N ,
\end{equation}
where the sum runs over all permutations of the string with $N-k$ spin down and $k$ spin up, and $\mathcal{N}$ is the 
normalization constant. The Dicke state $\ket{D_{N}^{(k)}}$ corresponds to $\ket{j,m}$ with $j=N/2$ and $m=k-N/2$.

Let $\mathcal{L(H)}$ be the Hilbert space of linear operators
acting on the finite-dimensional space $\mathcal{H}$. An operator basis for
$\mathcal{L(H)}$ equipped with the standard Hilbert-Schmidt inner product is given by the set of the $4^N$ generalized Pauli
matrices defined as the $N$-fold tensor products of the $2\times 2$
matrices $\sigma_0, \sigma_1, \sigma_2, \sigma_3$~\cite{Nie00}, 
\begin{equation}
\label{Ximatrices}
\boldsymbol{\sigma}_{\mu_1\mu_2\ldots\mu_N}=\sigma_{\mu_1}\otimes\sigma_{\mu_2}\otimes\ldots\otimes\sigma_{\mu_N}.
\end{equation}
These Hermitian operators verify the relations $\tr(\boldsymbol{\sigma}_{\mu_1\mu_2\ldots\mu_N}\boldsymbol{\sigma}_{\nu_1\nu_2\ldots\nu_N})=2^N\delta_{\mu_1 \nu_1}\delta_{\mu_2 \nu_2}\ldots \delta_{\mu_N \nu_N}$ and thus form an orthogonal basis. Any state $\rho$ of $N$ spin-$1/2$ can be expanded in this basis as  
\begin{equation}
\label{rhoarbitrary}
\rho=\frac{1}{2^N}\,x_{\mu_1\mu_2\ldots\mu_N} \boldsymbol{\sigma}_{\mu_1\mu_2\ldots\mu_N},
\end{equation}
where $x_{\mu_1\mu_2\ldots\mu_N}$ are real coefficients given by
\begin{equation}
\label{coeffdecompPauli}
x_{\mu_1\mu_2\ldots\mu_N} = \tr(\rho\,\boldsymbol{\sigma}_{\mu_1\mu_2\ldots\mu_N}).
\end{equation}
We can now prove the following theorem:
\begin{theorem}
\label{THM0}
The Weinberg covariant matrices defined in Eq.~(\ref{egaliteweinberg})
are given by the projection of tensor products of Pauli
matrices into the subspace $\mathcal{H}_S$ of states that are invariant under
permutation of particles. Namely, denoting by
$\mathcal{P}_S\equiv\sum_{k=0}^N\ket{D_N^{(k)}}\bra{D_N^{(k)}}$ the
projector onto $\mathcal{H}_S$, the $S_{\mu_1\mu_2\ldots\mu_N}$ matrix corresponds to the $(N+1)$-dimensional block spanned by the $\ket{D_N^{(k)}}$ of the matrix $\mathcal{P}_S\,\boldsymbol{\sigma}_{\mu_1\mu_2\ldots\mu_N}\,\mathcal{P}_S^\dagger$, i.e., in terms of matrix elements
\begin{equation}
\label{Squbit}
\langle D_N^{(k)}|S_{\mu_1\mu_2\ldots\mu_N}|D_N^{(\ell)}\rangle=\langle D_N^{(k)}|\boldsymbol{\sigma}_{\mu_1\mu_2\ldots\mu_N}|D_N^{(\ell)}\rangle,
\end{equation}
with $0\leqslant k,\ell\leqslant N$.
\end{theorem}
\begin{proof}\label{pr.thm0}
Let $\tilde{S}_{\mu_1\mu_2\ldots\mu_N}=\mathcal{P}_S\,\boldsymbol{\sigma}_{\mu_1\mu_2\ldots\mu_N}\,\mathcal{P}_S^\dagger$. Any spin-$j$ coherent state $|\alpha\rangle$ defined by Eq.~\eqref{coh} can also be written as the tensor product of identical
spin-1/2 coherent states. 
As a symmetric state, $|\alpha\rangle$ is invariant under $\mathcal{P}_S$,
i.e., $|\alpha\rangle=\mathcal{P}_S|\alpha\rangle$, so that 
$\langle\alpha|\tilde{S}_{\mu_1\mu_2\ldots\mu_N}|\alpha\rangle=\langle\alpha|\boldsymbol{\sigma}_{\mu_1\mu_2\ldots\mu_N}|\alpha\rangle=n_{\mu_1}n_{\mu_2}\ldots
n_{\mu_{N}}$. 
Using Eq.~\eqref{eq:Stn}, we thus have
\begin{equation}
  \label{eq:SStn}
\langle 
\alpha|\tilde{S}_{\mu_1\mu_2\ldots\mu_N}|\alpha\rangle=\langle \alpha|S_{\mu_1\mu_2\ldots\mu_N}|\alpha\rangle
\end{equation}
  for all $\alpha$, i.e., the Husimi functions of the two operators are
  identical. Therefore  $S_{\mu_1\mu_2\ldots\mu_N}$ and
  $\tilde{S}_{\mu_1\mu_2\ldots\mu_N}$ coincide
in $\mathcal{H}_S$. 
\end{proof}
In other words, instead of obtaining the Weinberg matrices from the
expansion of the rather complicated expressions
(\ref{pientier})--(\ref{pidemi}), we can construct them simply by
projecting the corresponding tensor product of Pauli operators into the
symmetric subspace.  In order to fully exploit the consequences of
this fact, we need some basic notions of frame theory~\cite{Cas13}. 

A family of vectors
  $|\phi_i\rangle$, $i\in\{1,\ldots,M\}$, is called a frame for a
  Hilbert space ${\cal H}$ with bounds $A, B\in ]0, \infty[$, if
\begin{equation}
    \label{eq:framedef}
    A||\psi||^2\leqslant \sum_{i=1}^M|\langle \psi|\phi_i\rangle|^2\leqslant
    B||\psi||^2, \;\;\;\forall\; |\psi\rangle\in{\cal H}. 
  \end{equation}
If $A=B$, then the frame is called an $A$-tight frame.

Orthonormal bases are a special case of $A$-tight frames. In particular, the
generalized Pauli matrices \eqref{Ximatrices} form -- up to normalization --
an orthonormal basis of $\mathcal{L(H)}$, and are in fact an $A$-tight frame, which verifies Eq.~(\ref{eq:framedef}) with $A=B=2^N$ and $M=4^N$. According to proposition 22 in~\cite{Cas13}, a frame of a Hilbert space
${\cal H}$ 
with bounds $A,B$ that is orthogonally projected to a subspace $P
{\cal H}$
is a frame  of $P{\cal H}$ with the same bounds $A,B$.  Therefore
we have as a corollary of Theorem
\ref{THM0} that the set of covariant matrices $S_{\mu_1\mu_2\ldots\mu_N}$ forms a
$2^N$-tight frame for $\mathcal{L}(\mathcal{H}_S)$.

Tight frames are in a sense a generalization of orthonormal bases, as they allow an expansion
over the elements of the frames with the same formulas as for an
orthonormal basis, i.e., for all 
$|\psi\rangle\in {\cal H}$, we have $|\psi\rangle=A^{-1}\sum_{i=1}^M\langle
\phi_i|\psi\rangle |\phi_i\rangle$ (proposition 20 in~\cite{Cas13}).
This immediately entails the following result, which
provides a generalization of 
the Bloch sphere representation for spin-1/2, Eq.~(\ref{canonrhoj12}), to any spin:
\begin{theorem}
\label{THM1}
For general spin-$j$, the $4^{N}$ Hermitian matrices
$S_{\mu_1\mu_2\ldots\mu_{N}}$ (with $N\equiv 2j$) provide an overcomplete basis (more precisely, a $2^N$-tight frame) over which $\rho$ can be expanded, that is, any
state can be expressed as 
\begin{equation}
\rho =\frac{1}{2^{N}}\,x_{\mu_1\mu_2\ldots\mu_{N}}S_{\mu_1\mu_2\ldots\mu_{N}}, 
\label{canonrhoj}
\end{equation}
with coefficients
\begin{equation}
\label{defcoor}
x_{\mu_1\mu_2\ldots\mu_{N}}=\tr(\rho\, S_{\mu_1\mu_2\ldots\mu_{N}})
\end{equation}
real and invariant under permutation of the indices.
\end{theorem}
Since $S_{00\ldots 0}$ is the identity matrix, the condition
$\tr\rho=1$ for density matrices is equivalent to $x_{00\ldots
  0}=1$. The tight frame property allows one to write the Hilbert-Schmidt scalar product of any two
Hermitian operators $\rho$ and $\rho'$ with coordinates
$x_{\mu_1\mu_2\ldots\mu_{N}}$ and  $x'_{\mu_1\mu_2\ldots\mu_{N}}$ as
the scalar product of coordinates, more precisely 
\begin{equation}
\label{scalarprod}
\tr(\rho\rho')=\frac{1}{2^N}x_{\mu_1\mu_2\ldots\mu_{N}}
x'_{\mu_1\mu_2\ldots\mu_{N}}.  
\end{equation}
The condition $\tr\rho^2\leqslant 1$ that every state must satisfy translates into $
\sum_{\mu_1}\ldots \sum_{\mu_N} x^2_{\mu_1\mu_2\ldots\mu_{N}}\leqslant 2^N$. Note that from Eq.~\eqref{defcoor}  and the definition
of $S_{\mu_1\mu_2\ldots\mu_N}$, the coordinates
$x_{\mu_1\mu_2\ldots\mu_N}$ appear as the coefficients of $(-1)^{N}\langle\Pi^{(j)}(q)\rangle$, which is a multivariate polynomial in variables $q_0, q_1,q_2, q_3$,
\begin{equation}
\label{eq:xq}
(-1)^{N}\langle\Pi^{(j)}(q)\rangle=x_{\mu_1\mu_2\ldots\mu_N}q_{\mu_1}\ldots q_{\mu_N}.
\end{equation}

Due to the overcompleteness of the $S_{\mu_1\mu_2\ldots\mu_{N}}$
the coordinates $x_{\mu_1\mu_2\ldots\mu_{N}}$ in \eqref{canonrhoj} are so far not unique. However, for a given spin-$j$ density matrix $\rho$, \eqref{defcoor} is the
unique choice of coordinates $x_{\mu_1\mu_2\ldots\mu_{N}}$ such that
these coordinates are real numbers, invariant under
permutation of the indices, and verifying the condition
$g_{\mu_1\mu_2}x_{\mu_1\mu_2\ldots\mu_{N}}=0$ (see Proposition 1 in the Supplemental Material).

The generalized Bloch representation (\ref{canonrhoj}) shares with the
Bloch representation of a spin-1/2 several crucial properties.  
First of all, using Eqs.~\eqref{eq:Stn} and \eqref{defcoor}, we see
that coordinates of a coherent state are simply given by the product
of components of the 4-vector $n=(1,{\bf n})$, namely
$x_{\mu_1\mu_2\ldots\mu_{N}}=n_{\mu_1}n_{\mu_2}\ldots
n_{\mu_{N}}$. This generalizes the fact that the Bloch vector
representing a spin-1/2 state points in 
the direction given by the angles defining the coherent state.
Secondly, under any SU(2) transformation, the Bloch vector of a
spin-1/2 simply rotates, i.e., transforms according to $x_a\to
R_{ab}x_b$, where $R$ is a rotation matrix.  Similarly, for higher spins
the tensor of coordinates of an arbitrary state transforms according to 
$x_{\mu_1\ldots \mu_N}\to R_{\mu_1\nu_1}\ldots
R_{\mu_N\nu_N}x_{\nu_1\ldots \nu_N}$, with $R_{ab}$ the $3\times 3$
rotation matrix and $R_{0\mu}=R_{\mu 0}=\delta_{\mu0}$. This is a
consequence of a more general covariance property of the basis
matrices $S_{\mu_1\mu_2\ldots\mu_N}$. Indeed, they were constructed in such a way that for any element $\Lambda$ of the Lorentz group, with
$D^{(j)}[\Lambda]$ the $(2j+1)$-dimensional 
matrix associated with $\Lambda$ in the $(j,0)$ representation,
\begin{equation} 
D^{(j)}[\Lambda]S_{\mu_1\mu_2\ldots\mu_N}D^{(j)}[\Lambda]^{\dagger}
=\Lambda^{\nu_1}_{\mu_1}\ldots \Lambda^{\nu_N}_{\mu_N}S_{\nu_1\nu_2\ldots\nu_N}
\end{equation}
in the covariant-contravariant notation of~\cite{Weinberg}. From
Eq.~\eqref{defcoor} this property translates to coordinates
$x_{\mu_1\mu_2\ldots\mu_{N}}$. For rotations $R_{\mu\nu}$, the distinction between upper and lower indices becomes irrelevant.

In addition to the shared advantages of a Bloch vector, 
our generalized Bloch sphere representation \eqref{canonrhoj} enjoys
additional convenient 
properties relevant for systems made of many spin-$1/2$ or qubits. For instance, coordinates of
the spin-$k$ reduced density matrix obtained by tracing the spin-$j$
matrix over $j-k$ spins are simply given by 
\begin{equation}
\label{proptrace}
x_{\mu_1\ldots \mu_{2k}}=x_{\mu_1\ldots \mu_{2k} 0\ldots0}
\end{equation}
(see Proposition 3 in the Supplemental Material). Note that in
\cite{Arn13} a similar property was observed for the coefficients in
the expansion of $\rho$ over generalized Pauli matrices, and a formal
Lorentz invariance of that expansion was used very recently to
generalize monogamy relations of entanglement~\cite{EltSie14}.\\

We now consider a few examples of states and give their coordinates in our representation. The maximally mixed state
$\rho_0=\frac{1}{2j+1}\mathbb{1}_{2j+1}$ has coordinates $x_{\mu_1\mu_2\ldots\mu_{2j}}$ given by
\begin{equation}
\label{coorid}
x_{\mu_1\mu_2\ldots\mu_{2j}}q_{\mu_1}\ldots q_{\mu_{2j}}=\sum_{k=0}^{j}\frac{\binom{2j}{2k}}{2k+1}q_0^{2(j-k)}|\bq|^{2k}
\end{equation}
(see Proposition 2 in the Supplemental Material). Another example is given by the Schr\"{o}dinger cat states $\ket{\psi_{\mathrm{cat}}^{(j)}}=(\ket{j,-j}+\ket{j,j})/\sqrt{2}$. By linearity of the expansion \eqref{canonrhoj} and of the trace, they have coordinates 
\begin{equation}
\label{coorschro}
x_{\mu_1\ldots \mu_{N}}^{\mathrm{cat}}=\frac{1}{2}\Bigg(\prod_{i=1}^{N} n_{\mu_i}^{(-\frac{1}{2},-\frac{1}{2})}+\prod_{i=1}^{N}n_{\mu_i}^{(\frac{1}{2},\frac{1}{2})}\Bigg)
+\mathrm{Re}\left[\prod_{i=1}^{N}n_{\mu_i}^{(-\frac{1}{2},\frac{1}{2})}\right]
\end{equation}
where $n^{(\pm \frac{1}{2},\pm \frac{1}{2})}=(1,0,0,\pm 1)$ are the
coordinates of the coherent states $\ket{\frac{1}{2},\pm
  \frac{1}{2}}\bra{\frac{1}{2},\pm \frac{1}{2}}$ and
$n^{(-\frac{1}{2},\frac{1}{2})}=(0,1,-i,0)$ are the coordinates of
the non-Hermitian operator $\ket{\frac{1}{2},-
  \frac{1}{2}}\bra{\frac{1}{2},\frac{1}{2}}$.

While the complete characterization of the set of coordinates for which $\rho$ is positive is difficult in any representation~\cite{SUNpos,Polapos,Weyl}, our representation \eqref{canonrhoj} allows one to solve this problem explicitely for $j=1$. The set of all spin-1 states is
characterized by 8 real parameters.  The transformation of tensor
$x_{\mu\nu}$ by rotation matrices under SU(2) operations allows one to
diagonalize the 3$\times$3 block $x_{ab}$ ($1\leqslant a,b\leqslant 3$), and
Eq.~(\ref{traceless}) imposes $\sum_{i=1}^3\mu_i=1$ for the eigenvalues
$\mu_i$, leaving five real parameters $\mu_1,\mu_2$,
$\bx\equiv(x_{01},x_{02},x_{03})$. In this case, $\bx$ coincides with $\bu$ in
the representation found in~\cite{GiraudBB}.  We therefore
immediately obtain that up to two special cases of measure zero the
set of all spin-1 states can be 
represented as a two--parameter family of ellipsoids in the space of
vectors $\bx$ (Eq.~(21) in~\cite{GiraudBB} with $\bu=\bx$ and $w_{ab}=x_{ab}$), thus providing a simple geometrical picture of all spin-1 states.

As a direct application of our formalism, we give a simple necessary
and sufficient criterion for anticoherence of spin states. 
Spin states are said to be anticoherent to order $t$ if $\langle
(\bn\boldsymbol{\cdot}\bJ)^k\rangle$ is independent on the unit vector $\bn$  for any
$k$ with $0\leqslant k\leqslant t$~\cite{Zimba}. Various characterisations have
been given~\cite{bannai}. Very recently the case of pure but not necessarily symmetric states was considered in~\cite{Arn13,GoyZyc14}. The definition of matrices $S_{\mu_1\mu_2\ldots\mu_{N}}$ via \eqref{lorentzboost}--\eqref{egaliteweinberg} as a function of $\bJ$ makes them most convenient for the characterisation of anticoherent states. One can show the following result:
\begin{theorem}
\label{th:anticoherence}
A spin-$j$ state $\rho$, pure or mixed, is anticoherent to order $t$ if and only if its spin-$(t/2)$ reduced density matrix is the maximally mixed state $\rho_0=\frac{1}{t+1}\mathbb{1}_{t+1}$.
\end{theorem}

The proof (see Supplemental Material for more detail) relies on the calculation of $\langle \Pi^{(j)}(q)\rangle$ for an anticoherent state, using the expansion \eqref{pientier}--\eqref{pidemi} and identifying terms up to order $t$ with the expansion \eqref{coorid} of the maximally mixed state. For instance, spin-$j$ anticoherent states to order 1 are characterized by $\langle S_{\mu 00\ldots 0}
\rangle=\delta_{\mu 0}$ while anticoherent states to order 2 are characterized by $\langle S_{\mu \nu 00\ldots 0}
\rangle=\mathrm{diag}(1,1/3,1/3,1/3)$. From the characterization of anticoherence given by Theorem \ref{th:anticoherence}, one can easily obtain another characterization based on coefficients of the multipolar expansion of the density matrix. For a spin-$j$ density operator $\rho$, the expansion reads
\begin{equation}
\label{rhomultipolar}
\rho=\sum_{k=0}^{2j}\sum_{q=-k}^k \rho_{kq}\,T_{kq}^{(j)}
\end{equation}
with $\rho_{kq}=\tr\big(\rho\, {T_{kq}^{(j)}}^{\dagger}\big)$, where $T_{kq}^{(j)}$ are the irreducible tensor operators~\cite{Aga81}
\begin{equation}
\label{irrecT}
T_{kq}^{(j)}=\sqrt{\frac{2k+1}{2j+1}}\sum_{m,m'=-j}^j C_{jm,kq}^{jm'} \ket{j,m'}\bra{j,m},
\end{equation}
and $C_{jm,kq}^{jm'}$ are Clebsch-Gordan coefficients. The following corollary of Theorem 3 can now be stated (see Supplemental Material for a proof).

{\bf Corollary 1.} \emph{A spin-$j$ state $\rho$ is anticoherent to order $t$ if and only if $\rho_{kq}=0,\;\;\forall\;k\leqslant t$, $\forall\;q:-k\leqslant q\leqslant k$.}

Note that in~\cite{BBM14} the current characterizations were obtained up to second order.

In summary, we have introduced a tensorial representation of spin states that leads to a natural generalization of the Bloch sphere
representation to arbitrary spin $j$, based on Weinberg's covariant
matrices~\cite{Weinberg}.  We have found a convenient way of representing these matrices as
projections of elements of the Pauli group into the symmetric subspace
of $2j$ spins-1/2, proving that they form
a tight frame. Our representation shares 
beautiful and essential properties with the one for spin-1/2 (or qubit), and
provides additional insight 
for larger spins that we have used for a novel characterization of
anticoherent spin states. We expect that the mathematical elegance of our representation will enable new insights in different fields of physics where spins are relevant.

\end{document}


\title{Supplementary Information}
\maketitle

In this supplemental material, we provide the proofs of the propositions, lemmas and theorems 1 and 3 stated in the main text.
We recall that covariant matrices $S_{\mu_1\mu_2\ldots\mu_{2j}}$ are obtained by identifying coefficients of 
\begin{equation}
\label{egaliteweinberg}
\Pi^{(j)}(q)=(-1)^{2j}q_{\mu_1}q_{\mu_2}\ldots q_{\mu_{2j}}S_{\mu_1\mu_2\ldots\mu_{2j}}
\end{equation}
with the expansion of the polynomials (multivariate in the $q_{\mu}$) 
\begin{equation}
\label{pientier}
\Pi^{(j)}(q)=(q_0^2-{\bf q}^2)^j+\sum_{k=1}^{j}\frac{(q_0^2-{\bf q}^2)^{j-k}}{(2k)!}(\dqJ)\Bigg(\prod_{r=1}^{k-1}[(\dqJ)^2-(2r\bq)^2]\Bigg)[\dqJ+ 2k q_0]
\end{equation}
for integer $j$, and
\begin{equation}
\label{pidemi}
\Pi^{(j)}(q)=(q_0^2-{\bf q}^2)^{j-1/2}(-q_0-\dqJ)-\sum_{k=1}^{j-1/2}\frac{(q_0^2-{\bf q}^2)^{j-1/2-k}}{(2k+1)!}\left(\prod_{r=1}^{k}[(\dqJ)^2-((2r-1)\bq)^2]\right)(\dqJ+(2k+1)q_0)
\end{equation}
for half-integer $j$, with $q=(q_{0},q_{1},q_{2},q_{3})\equiv (q_0,\bq)$. We denote $N=2j$. The operator $\Pi^{(j)}(q)$ is proportional to the square of the Hermitian operator associated to a Lorentz boost in direction ${\bf q}$ for a particle of mass $m$,
\begin{equation}
\label{lorentzboost}
\Pi^{(j)}(q)=m^{N}\exp\left(-2\eta_q\,\hbq\boldsymbol{\cdot}\bJ\right), 
\end{equation}
where $\hbq=\bq/|\bq|$, and $\eta_q$ and $m$ are defined by 
\begin{eqnarray}
\label{defq0}
q_{0}&=&-m\cosh\eta_q,\\
|\bq|&=&m\sinh\eta_q\ .
\label{defqq}
\end{eqnarray}
We recall that the $S_{\mu_1\mu_2\ldots\mu_{N}}$ are linked by a linear relation, given by
\begin{equation}
\label{traceless}
g_{\mu_1\mu_2}S_{\mu_1\mu_2\ldots\mu_{N}}=0,
\end{equation}
where $g\equiv\mathrm{diag}(-,+,+,+)$.
Theorem 2 which was proved in the paper states that for general spin--$j$, the Hermitian matrices $S_{\mu_1\mu_2\ldots\mu_N}$ provide an overcomplete basis over which $\rho$ can be expanded, that is, any state can be expressed as
\begin{equation}
\rho =\frac{1}{2^N}x_{\mu_1\mu_2\ldots\mu_N}S_{\mu_1\mu_2\ldots\mu_N}, 
\label{canonrhoj}
\end{equation}
with real coefficients given by 
\begin{equation}
\label{defcoor}
x_{\mu_1\mu_2\ldots\mu_N}=\tr(\rho S_{\mu_1\mu_2\ldots\mu_N}).
\end{equation}

We now turn to the proofs.
\begin{lemma}
\label{lemmedaniel}
Let $\ket{\cohalpha}$ be a spin--$j$ coherent state, defined for $\alpha=e^{-i\varphi}\cot(\theta/2)$ with
$\theta\in[0,\pi]$ and $\varphi\in [0,2\pi[$ by
\begin{eqnarray}
\label{coh}
\ket{\cohalpha}&=&\sum_{m=-j}^{j}\sqrt{\binom{2j}{j+m}}\left[\sin\!\tfrac{\theta}{2}\right]^{j-m}\left[\cos\!\tfrac{\theta}{2}e^{-i\varphi}
\right]^{j+m}\ket{j,m},\\
&=&\frac{1}{(1+|\alpha|^2)^j}e^{\alpha J_+}|j,-j\rangle, \label{coh2}
\end{eqnarray}
in the standard angular momentum basis $\{\ket{j,m}:-j\leqslant m \leqslant j\}$, and let $\bn=(\sin\theta\cos\varphi,\sin\theta\sin\varphi,\cos\theta)$. Then
\begin{equation}
\label{eq:prop}
\bra{\cohalpha}\Pi^{(j)}(q)\ket{\cohalpha}=(-1)^{N}(q_0+\bq\boldsymbol{\cdot}\bn)^{N}.
\end{equation}
\end{lemma}
\begin{proof}
We use the SU(2) disentangling theorem~[33]: for $a_+, a_-,b\in\mathbb{C}$,
\begin{equation}
\label{eq:decon}
\exp(a_+J_++bJ_z+a_-J_-)=e^{b_-J_-}e^{b_z J_z}e^{b_+J_+}
\end{equation}
with constants given by
\begin{eqnarray}
\label{eq:bs}
b_{\pm}&=&2a_{\pm}\frac{\sinh (\delta/2)}{\delta \cosh (\delta/2)+b\sinh(\delta/2)}\\ 
\label{eq:bs2}
b_z&=&2\ln\left(\frac{\delta\cosh(\delta/2)+b\sinh(\delta/2)}{\delta}\right),
\end{eqnarray}
where $\delta=\sqrt{b^2+4a_+a_-}$. From its definition \eqref{lorentzboost}, $\Pi^{(j)}(q)$ can be written as in \eqref{eq:decon} with $a_{\pm}=-\eta_q(\hat{q}_x\mp i \hat{q}_y)$ and $b=-2\eta_q\hat{q}_z$. Thus, using \eqref{coh2} and \eqref{eq:decon} we get
\begin{equation}
  \label{eq:L}
\bra{\cohalpha}\Pi^{(j)}(q)\ket{\cohalpha}=m^{N}\bra{\psi}e^{b_z J_z}\ket{\psi}
\end{equation}
with 
\begin{equation}
\label{eq:psi}
\ket{\psi}=\frac{1}{(1+|\alpha|^2)^j}e^{(\alpha+b_+)J_+}\ket{j,-j}.
\end{equation}
If we parametrize $\hbq=(\sin\gamma\cos\phi,\sin\gamma\sin\phi,\cos\gamma)$, we get from \eqref{eq:bs}--\eqref{eq:bs2} that the constants are given by $\delta=2\eta_q$ and
\begin{eqnarray}
b_{\pm}&=&-\frac{\sin\gamma e^{\mp i\phi}}{\coth\eta_q-\cos\gamma},\\
b_z&=&2\ln\left(\cosh\eta_q -\sinh\eta_q\cos\gamma\right).
\end{eqnarray}
We see that, up to normalization, $\ket{\psi}$ is of the form \eqref{coh2}. Namely, we have
\begin{equation}
\ket{\psi}=\left(\frac{1+|\alpha+b_+|^2}{1+|\alpha|^2}\right)^j\ket{\alpha+b_+}.
\end{equation}
From Eq.~\eqref{coh} we get for a coherent state $\ket{\cohalpha}$ and a complex number $t$ the identity
\begin{equation}
\label{eq:moycoh}
\bra{\cohalpha}e^{2 t J_z}\ket{\cohalpha}=\left(\cosh t+\sinh t\cos\theta\right)^{N}\,.
\end{equation}
Applying \eqref{eq:moycoh} to the coherent state $\ket{\alpha+b_+}$ gives
\begin{equation}\label{aPia}
\bra{\cohalpha}\Pi^{(j)}(q)\ket{\cohalpha}=m^{N}\left(\frac{1+|\alpha+b_+|^2}{1+|\alpha|^2}\right)^{N}\left(\cosh\frac{b_z}{2}+\frac{|\alpha+b_+|^2-1}{|\alpha+b_+|^2+1}\sinh\frac{b_z}{2}\right)^{N}.
\end{equation}
On the right-hand side of (\ref{eq:prop}) we have the $N$th power of
\begin{equation}
  \label{eq:rhs}
-q_0-\bq\cdot\mathbf{n}=m\left(\cosh\eta_q-\sinh\eta_q(\sin\theta\sin\gamma\cos(\varphi-\phi)+\cos\gamma\cos\theta)\right)\,. 
\end{equation}
Using the explicit expressions of $\alpha$, $b_+$ and $b_z$ in Eq.~\eqref{aPia}, Eq.~\eqref{eq:prop} is now equivalent to a trigonometric identity that is easily verified.
 \end{proof}


\begin{proposition}
For a given density matrix $\rho$, Eq.~\eqref{defcoor} is the unique
choice of coordinates which are real numbers symmetric under permutation of the
indices and verifying 
\begin{equation}
\label{traceless_coor}
g_{\mu_1\mu_2}x_{\mu_1\mu_2\ldots\mu_{N}}=0
\end{equation}
\end{proposition}
\begin{proof}
The fact that Eq.~\eqref{traceless_coor} is verified follows immediately from Eqs.~\eqref{traceless} and \eqref{defcoor}, and by linearity of the trace. It remains to show the uniqueness of this choice. Matrices $S_{\mu_1\mu_2\ldots\mu_{N}}$  are not linearly independent
since they are invariant under permutation of indices and related
through Eq.~\eqref{traceless}. The number of distinct sets
$(\mu_1,\ldots,\mu_{N})$ up to permutation of indices is
$\binom{N+3}{3}$, from which one has to subtract the  $\binom{N+1}{3}$
 relations \eqref{traceless}. That leaves 
\begin{equation}
\binom{N+3}{3}-\binom{N+1}{3}=(N+1)^2
\end{equation}
independent basis matrices, which coincides with the number of
independent parameters in a $(N+1)\times (N+1)$ Hermitian
matrix (here we disregard the fact that $\tr\rho=1$, which would just
correspond to imposing $x_{00\ldots 0}=1$). This means that there
cannot be any other relation between the $S_{\mu_1\mu_2\ldots\mu_{N}}$
than the permutation-symmetry relations and the relations \eqref{traceless} (otherwise there would not be enough parameters to describe
all matrices $\rho$). In particular, if we fix $\mu_3,\ldots,\mu_{N}$,
the vector space $V_{\mu_3\ldots\mu_{N}}$ generated by matrices
$S_{\nu\nu\mu_3\ldots\mu_{N}},0\leqslant\nu\leqslant 3$, is of dimension $\leqslant
4$, and from Eq.~\eqref{traceless} we get that $V_{\mu_3\ldots\mu_{N}}$ is of dimension $3$. So the only
additional relation one can impose between the
$x_{\mu_1\mu_2\ldots\mu_{N}}$ is between the
$x_{\nu\nu\mu_3\ldots\mu_{N}}$. Choosing Eq.~\eqref{traceless_coor}
thus defines coordinates of a density matrix in a unique way. 
\end{proof}

\begin{proposition}
\label{prop:coorid}
Coordinates of the maximally mixed state $\rho_0=\frac{1}{N+1}\mathbb{1}_{N+1}$ (with $\mathbb{1}_{N+1}$ the $(N+1)$-dimensional identity matrix) are given by
\begin{equation}
\label{coorid}
x_{\mu_1\mu_2\ldots\mu_{N}}q_{\mu_1}\ldots q_{\mu_{N}}=\sum_{k=0}^{j}\frac{\binom{N}{2k}}{2k+1}q_0^{2(j-k)}|\bq|^{2k}. 
\end{equation}
\end{proposition}
\begin{proof}
We use the expansion of the identity in terms of coherent states,
\begin{equation}
\label{decomprho0}\rho_0=\frac{1}{4\pi}\int d\cohalpha\ket{\cohalpha}\bra{\cohalpha}.
\end{equation}
According to the main text, the coordinates $x_{\mu_1\mu_2\ldots\mu_{N}}$ are the coefficients of the multivariate polynomial $(-1)^{N}\langle\Pi^{(j)}(q)\rangle$. Using \eqref{decomprho0} we get
\begin{eqnarray}
\tr(\rho_0\Pi^{(j)}(q))&=&\frac{1}{4\pi}\int d\cohalpha\bra{\cohalpha}\Pi^{(j)}(q)\ket{\cohalpha}\nonumber\\
&=&(-1)^{N}\frac{1}{4\pi}\int d\cohalpha (q_0+\bq\boldsymbol{\cdot}\bn)^{N},
\end{eqnarray}
where the last equality comes from Lemma \ref{lemmedaniel}. Since the integral runs over the whole sphere, one can take $\bq=(0,0,|\bq|)$ and rewrite the integral as
\begin{equation}
\begin{aligned}
&\frac{(-1)^{N}}{2}\int_0^{\pi} d\theta (q_0+|\bq|\cos\theta)^{N}\sin\theta\\
&=(-1)^{N}\sum_{k=0}^{N}\binom{N}{k}q_0^{N-k}|\bq|^k\frac{1+(-1)^k}{2k+2}.
\end{aligned}
\end{equation}
Using Eq.~\eqref{egaliteweinberg}, we get the result.
\end{proof}

\begin{proposition}
\label{proptrace}
Coordinates of the spin--$k$ reduced density matrix obtained by tracing the spin--$j$ matrix over $j-k$ spins are given by
\begin{equation}
x_{\mu_1\ldots \mu_{2k}}=x_{\mu_1\ldots \mu_{2k} 0\ldots0}
\end{equation}
\end{proposition}
\begin{proof}
Let $\rho$ be a density matrix. Expanding $\rho$ over coherent states as
\begin{equation}
\label{rhoascoh}
\rho=\int d\cohalpha P(\cohalpha)\ket{\cohalpha}\bra{\cohalpha},
\end{equation}
we get 
\begin{equation}
\label{eqaux}
\tr(\rho\Pi^{(j)}(q))=\int d\cohalpha P(\cohalpha)\bra{\cohalpha}\Pi^{(j)}(q)\ket{\cohalpha}.
\end{equation}
Using Lemma \ref{lemmedaniel} and the expansion \eqref{egaliteweinberg}, Eq.~\eqref{eqaux} gives
\begin{equation}
\label{eqaux2}
q_{\mu_1}\ldots q_{\mu_{2j}}\tr(\rho S_{\mu_1\mu_2\ldots\mu_{2j}})=\int d\cohalpha P(\cohalpha)(q_0+\bq\boldsymbol{\cdot}\bn)^{2j},
\end{equation}
where $\bn$ is the unit vector associated with $\ket{\alpha}$. In the expansion \eqref{rhoascoh}, coherent states $\ket{\alpha}$ are a $2j$--fold tensor product of spin--$1/2$ coherent states $\ket{\alpha^{(1/2)}}$. The trace of $\ket{\cohalpha}\bra{\cohalpha}$ over $j-k$ spins is the $2k$--fold tensor product of the projector on $\ket{\alpha^{(1/2)}}$. From Eqs.~\eqref{eqaux2} and \eqref{defcoor}, we get
\begin{equation}
\label{coeff2j}
x_{\mu_1\mu_2\ldots\mu_{2j}}q_{\mu_1}\ldots q_{\mu_{2j}}=\int d\cohalpha P(\cohalpha)(q_0+\bq\boldsymbol{\cdot}\bn)^{2j}.
\end{equation}
and the coordinates $x_{\mu_1\mu_2\ldots\mu_{2k}}$ of the spin--$k$ reduced density matrix are thus given by
\begin{equation}
\label{coeff2k}
x_{\mu_1\mu_2\ldots\mu_{2k}}q_{\mu_1}\ldots q_{\mu_{2k}}=\int d\cohalpha P(\cohalpha)(q_0+\bq\boldsymbol{\cdot}\bn)^{2k}.
\end{equation}
The $x_{\mu_1\mu_2\ldots\mu_{2k}}$ defined by \eqref{coeff2k} can then be directly read off Eq.~\eqref{coeff2j}, which yields the result.
\end{proof}

\noindent{\bf Theorem 3.}
\emph{A spin--$j$ density matrix is anticoherent to order $t$ if and only if its spin--$(t/2)$ reduced density matrix is the maximally mixed state $\rho_0=\frac{1}{t+1}\mathbb{1}_{t+1}$.}

\begin{proof}
Let us first consider the case $t=2j$. The matrix $\rho_0$ is entirely characterized by the quantities $\langle\beta|\rho_0|\beta\rangle$, where $\ket{\beta}$ runs over coherent states. If one expands $\rho_0$ as in \eqref{rhoascoh}, then 
 \begin{equation}
 \label{condrho0}
 \langle\beta|\rho_0|\beta\rangle=\frac{1}{N+1}=\int d\cohalpha P(\cohalpha)(1+\bn'\boldsymbol{\cdot}\bn)^N,
 \end{equation}
where  $\bn'$ denotes the unit vector corresponding to  $\ket{\beta}$. The maximally mixed state is thus characterized by the fact that for any $\bn'$ its $P$--function verifies the right-hand equality in Eq.~\eqref{condrho0}. Any state such that the right-hand side of \eqref{condrho0} is independent of $\bn'$ is thus proportional to $\rho_0$.

 Anticoherence to order $2j$ means that $\langle (\bn\boldsymbol{\cdot}\bJ)^k\rangle$ is independent of $\bn$ for $k$ up to $2j$. The operator $\bn\boldsymbol{\cdot}\bJ$ has eigenvalues $-j, -j+1,\ldots,j$, and since its characteristic polynomial is also a minimal polynomial, one has 
\begin{equation}
\label{characpol}
\prod_{m=-j}^{j}(\bn\boldsymbol{\cdot}\bJ-m)=0,
\end{equation}
which, by the way, is the reason why \eqref{lorentzboost} can be expanded into a finite sum \eqref{pientier}--\eqref{pidemi}. From Eq.~\eqref{characpol} this means that $\langle (\bn\boldsymbol{\cdot}\bJ)^k\rangle$ is in fact independent of $\bn$ for any $k$, which in turn implies that $\langle\Pi^{(j)}(q)\rangle$ is independent of $\hbq$. If $P$ is the $P$--function associated with $\rho$ as in \eqref{rhoascoh}, Eqs.~\eqref{eqaux}--\eqref{eqaux2} imply that
\begin{equation}
\label{conditionrho0}
\int d\cohalpha P(\cohalpha)(q_0+\bq\boldsymbol{\cdot}\bn)^N
\end{equation}
is independent of $\hbq$. In particular, for $q_0=1$ and $\hbq=\bn'$ one recovers the condition \eqref{condrho0} which characterizes $\rho_0$ (up to a multiplicative constant, which is  then fixed by the normalization condition $\tr\rho_0=1$). Thus anticoherence to order $2j$ implies that $\rho=\rho_0$. The converse is true: since from Proposition \ref{prop:coorid} the coordinates of $\rho_0$ do not depend on $\hbq$, Eq.~(23) of the paper,
\begin{equation}
\label{eq:xq}
(-1)^{N}\langle\Pi^{(j)}(q)\rangle=x_{\mu_1\mu_2\ldots\mu_N}q_{\mu_1}\ldots q_{\mu_N},
\end{equation}
implies that $\langle\Pi^{(j)}(q)\rangle$ does not depend on $\hbq$, and thus the coefficients of its series expansion in powers of $\eta_q$, obtained from \eqref{lorentzboost}, do not depend on $\hbq$ either.

Let us now consider a state $\rho$ anticoherent to order $t\leqslant 2j$, and let $x_{\mu_1\ldots\mu_N}$ be its coordinates. From \eqref{defcoor} and proposition \ref{proptrace}, the spin--$(t/2)$ reduced density matrix of $\rho$ has coordinates given by $x_{\mu_1\ldots\mu_t 0\ldots 0}$. Following Proposition \ref{prop:coorid}, we thus want to show that  a state is anticoherent to order $t$ if and only if its coordinates are such that
\begin{equation}
\label{aux:th2a}
x_{\mu_1\mu_2\ldots\mu_{t}0\ldots 0}q_{\mu_1}\ldots q_{\mu_{t}}=\sum_{k=0}^{t/2}\frac{\binom{t}{2k}}{2k+1}q_0^{t-2k}|\bq|^{2k},
\end{equation}
or, equivalently,
\begin{equation}
\label{aux:th2b}
x_{\mu_1\mu_2\ldots\mu_{t}0\ldots 0}q_{\mu_1}\ldots q_{\mu_{t}}q_0^{N-t}=\sum_{k=0}^{t/2}\frac{\binom{t}{2k}}{2k+1}q_0^{N-2k}|\bq|^{2k}.
\end{equation}
From the form \eqref{pientier}, expanding the powers of $q_0^2-{\bf q}^2$, one has, for integer $j$,
\begin{equation}
\label{pientierproof}
\langle\Pi^{(j)}(q)\rangle=\sum_{s=0}^{j}\binom{j}{s}q_0^{2(j-s)}|\bq|^{2s}+\sum_{k=1}^{j}\sum_{i=0}^{j-k}\binom{j-k}{i}\frac{q_0^{2(j-k-i)}|\bq|^{2i}}{(2k)!}
\Big\langle(\dqJ)\left(\prod_{r=1}^{k-1}[(\dqJ)^2-(2r\bq)^2]\right)(\dqJ+ 2k q_0)\Big\rangle.
\end{equation}
Following Eq.~\eqref{eq:xq},  the $x_{\mu_1\mu_2\ldots\mu_{t}0\ldots 0}$ are obtained by considering the terms containing a factor $q_0^k$ with $k\geq N-t$ in $(-1)^N\langle\Pi^{(j)}(q)\rangle$. The term in $q_0^{2(j-s)}$ in \eqref{pientierproof} reads
\begin{equation}
\label{pientierproof2a}
q_0^{2(j-s)}\left\{\binom{j}{s}|\bq|^{2s}+\sum_{k=1}^{s}\binom{j-k}{s-k}\frac{|\bq|^{2(s-k)}}{(2k)!}
\Big\langle(\dqJ)\left(\prod_{r=1}^{k-1}[(\dqJ)^2-(2r\bq)^2]\right)(\dqJ)\Big\rangle\right\}
\end{equation}
and the term in $q_0^{2(j-s)+1}$ in \eqref{pientierproof} reads
\begin{equation}
\label{pientierproof2b}
q_0^{2(j-s)+1}\left\{\sum_{k=1}^{s}\binom{j-k}{s-k}\frac{|\bq|^{2(s-k)}}{(2k-1)!}
\Big\langle(\dqJ)\left(\prod_{r=1}^{k-1}[(\dqJ)^2-(2r\bq)^2]\right)\Big\rangle\right\}.
\end{equation}
The largest power of $\bq\boldsymbol{\cdot}\bJ$ in \eqref{pientierproof2a} corresponds to $k=s$ and $r=k-1$, which gives a power  $(\bq\boldsymbol{\cdot}\bJ)^{2s}$. From the definition of an anticoherent state of order $t$, we have that for $0\leqslant s\leqslant t/2$, all powers  of $\bq\boldsymbol{\cdot}\bJ$ appearing in \eqref{pientierproof2a} are such that their average does not depend on $\hbq$. A similar reasoning holds for Eq.~\eqref{pientierproof2b}. One can thus rewrite \eqref{pientierproof2a} and \eqref{pientierproof2b} respectively as
\begin{equation}
\label{pientierproof3a}
q_0^{N-2s}|\bq|^{2s}\sum_{k=0}^{s}\frac{\binom{j-k}{s-k}}{(2k)!}\Big\langle\prod_{r=0}^{k-1}[(2\hbq\boldsymbol{\cdot}\bJ)^2-(2r)^2]\Big\rangle
\end{equation}
and
\begin{equation}
\label{pientierproof3b}
q_0^{N-2s+1}|\bq|^{2s-1}\sum_{k=1}^{s}\frac{\binom{j-k}{s-k}}{(2k-1)!}
\Big\langle(2\hbq\boldsymbol{\cdot}\bJ)\prod_{r=1}^{k-1}[(2\hbq\boldsymbol{\cdot}\bJ)^2-(2r)^2]\Big\rangle,
\end{equation}
where any $\langle(\hbq\boldsymbol{\cdot}\bJ)^{k}\rangle$ can be replaced by $\langle(J_z)^{k}\rangle$. 

When summing over $\mu_1,\ldots,\mu_t$ in the left-hand side of  \eqref{aux:th2b}, the coefficient of a given $q_{\mu_1}\ldots q_{\mu_{t}}q_0^{N-t}$  is $x_{\mu_1\mu_2\ldots\mu_{t}0\ldots 0}$ multiplied by the number of permutations of $\{\mu_1\mu_2,\ldots,\mu_{t}\}$, while the coefficient of $q_{\mu_1}\ldots q_{\mu_{t}}q_0^{N-t}$ in \eqref{pientierproof}, or in \eqref{pientierproof3a}--\eqref{pientierproof3b}, is obtained from \eqref{egaliteweinberg} and \eqref{defcoor} as $x_{\mu_1\mu_2\ldots\mu_{t}0\ldots 0}$ multiplied by the number of permutations of $\{\mu_1,\mu_2,\ldots,\mu_{t},0,\ldots,0\}$. Let us group terms according to the number $p_{\nu}$ of indices $\{\mu_1,\mu_2,\ldots,\mu_{t}\}$ equal to $\nu$. In order to identify terms containing a power $q_0^{p_0+N-t}$ on both sides of Eq.~\eqref{aux:th2b} we must consider terms such that $k=(t-p_0)/2$ when $t-p_0$ is even (the contribution is 0 for odd $t-p_0$). Similarly we must take terms such that $s=(t-p_0)/2$ in \eqref{pientierproof3a} when $t-p_0$ is even, and terms such that $s=(t-p_0+1)/2$ in \eqref{pientierproof3b} when $t-p_0$ is  odd. In order to show equality between the $x_{\mu_1\mu_2\ldots\mu_{t}0\ldots 0}$ defined from \eqref{aux:th2b} and those defined from \eqref{pientierproof3a} for $t-p_0$ even, one must show that the coefficient of $q_0^{N-2s}|\bq|^{2s}$ in \eqref{pientierproof3a} with $s=(t-p_0)/2$ is equal to the coefficient of $q_0^{N-2k}|\bq|^{2k}$ with $k=(t-p_0)/2$, multiplied by
\begin{equation}
\label{combifac}
\frac{\#\textrm{perm.}\{\mu_1\mu_2\ldots\mu_{t}0\ldots 0\}}{\#\textrm{perm.}\{\mu_1\mu_2\ldots\mu_{t}\}}=
\frac{\binom{N}{p_0+N-t,p_1,p_2,p_3}}{\binom{t}{p_0,p_1,p_2,p_3}}=\frac{N!p_0!}{(p_0+N-t)!t!}=\frac{\binom{N}{2k}}{\binom{t}{2k}},
\end{equation}
where $\binom{n}{n_0,n_1,\ldots,n_k}$ stands for the multinomial coefficient $\frac{n!}{n_0!n_1!\,\ldots n_k!}$. From the right-hand side of  \eqref{aux:th2b}, the coefficient of $q_0^{N-2k}|\bq|^{2k}$ multiplied by the combinatorial factor \eqref{combifac} is readily seen to be 
\begin{equation}
\label{cooridk}
\frac{\binom{N}{2k}}{2k+1}
\end{equation}
for $(t-p_0)/2$ integer. The equality between \eqref{cooridk} and the coefficient obtained from \eqref{pientierproof3a} has to be shown for all $s$ from 0 to $t/2$. But note that $t$ has been eliminated both from \eqref{pientierproof3a} and  \eqref{cooridk}, so that for fixed $s$ or $k$ these equations are the same as the ones between coefficients of a spin--$j$ state anticoherent to order $2j$ (the only difference being that in the latter case Eq.~\eqref{pientierproof3a} would hold for all $s$ up to $s=j$). Since the case $t=2j$ has already been proved, the result ensues. The same argument applies to the case $t-p_0$ odd.  Finally, the result for half-integer $j$ can be derived along the same line of reasoning.
\end{proof}

{\bf Corollary 1.} \emph{A spin--$j$ state $\rho$ is anticoherent to order $t$ if and only if $\rho_{kq}=0,\;\;\forall\;k\leqslant t$, $\forall\;q: -k\leqslant q \leqslant k$.}
\begin{proof}

Let $\rho$ be a spin--$j$ density matrix with multipolar expansion
\begin{equation}
\label{rhomultipolar}
\rho=\sum_{k=0}^{2j}\sum_{q=-k}^k \rho_{kq}\,T_{kq}^{(j)}
\end{equation}
with $\rho_{kq}=\tr\big(\rho\, {T_{kq}^{(j)}}^{\dagger}\big)$, where $T_{kq}^{(j)}$ are the irreducible tensor operators
\begin{equation}
\label{irrecT}
T_{kq}^{(j)}=\sqrt{\frac{2k+1}{2j+1}}\sum_{m,m'=-j}^j C_{jm,kq}^{jm'} \ket{j,m'}\bra{j,m},
\end{equation}
and $C_{jm,kq}^{jm'}$ are Clebsch-Gordan coefficients. For an anticoherent state to order $t$, $\rho_{kq}=0\;\;\forall\;k\leqslant t$ follows from the fact that coefficients $\rho_{kq}$ are proportional to coefficients $R_{kq}$ of the expansion of $\langle\alpha|\rho|\alpha\rangle$ over spherical harmonics (see e.g.~[20]). From Eq.~(25) of the paper,
\begin{equation}
\label{scalarprod}
\tr(\rho\rho')=\frac{1}{2^N}x_{\mu_1\mu_2\ldots\mu_{N}}
x'_{\mu_1\mu_2\ldots\mu_{N}},
\end{equation}
 one gets
\begin{equation}
\label{sphericalY}
\langle\alpha|\rho|\alpha\rangle=\sum_{k,q}R_{kq}Y_{kq}(\theta,\varphi)
=\frac{1}{2^N}x_{\mu_1\mu_2\ldots\mu_{N}}n_{\mu_1}n_{\mu_2}\ldots n_{\mu_{N}}
\end{equation}
with $\bn=(\sin\theta\cos\varphi,\sin\theta\sin\varphi,\cos\theta)$.

If a state $\rho$ is such that its spin--$(t/2)$ reduced density
matrix is the maximally mixed state then from \eqref{proptrace} its
coefficients $x_{\mu_1\mu_2\ldots\mu_{t}0\ldots 0}$ are given by
\eqref{coorid}, and all terms $x_{\mu_1\mu_2\ldots\mu_{t}0\ldots
  0}n_{\mu_1}\ldots n_{\mu_t}$ on the right-hand side of
\eqref{sphericalY} can be resummed to a constant independent of
$\theta$ and $\varphi$ given by
\begin{equation}
x_{\mu_1\mu_2\ldots\mu_{t}0\ldots
  0}n_{\mu_1}\ldots n_{\mu_t}=\sum_{k=0}^{t/2}\frac{\binom{t}{2k}}{2k+1}n_0^{t-2k}|\mathbf{n}|^{2k}=\frac{2^t}{t+1}.
\end{equation}

What remains in the sum \eqref{sphericalY} are the terms $x_{\mu_1\mu_2\ldots\mu_{N}}n_{\mu_1}n_{\mu_2}\ldots n_{\mu_{N}}$ with at least $N-t$ non-vanishing indices $\mu_i$, yielding trigonometric polynomials in $\theta$ and $\varphi$ of order at least $t+1$ and thus all coefficients $R_{k\ell}$ with $k\leqslant t$ (apart from $R_{00}$) vanish, and so do coefficients $\rho_{kq}$.\\
\end{proof}

As an illustration, let us consider specific examples. 

\textbf{Spin--1 anticoherent states to order 1.} 
The most general form of spin--1 states that are anticoherent to order 1 is obtained by setting $x_{00}=1,\, x_{01}=x_{02}=x_{03}=0$. The condition $x_{00}=1$ imposes unit trace of $\rho$. The remaining conditions $x_{01}=x_{02}=x_{03}=0$ imply, according to Eqs.~\eqref{coorid} and \eqref{proptrace}, that the spin--$1/2$ reduced density matrix is maximally mixed. This yields a density matrix in the ${\ket{j,m}}$ basis of the form
\begin{equation}
\label{rho1}
\rho=\left(\begin{array}{ccc}
\frac{1}{2}+a & \beta & \gamma\protect\\
\beta^{*} & -2a & -\beta\protect\\
\gamma^{*} & -\beta^{*} & \frac{1}{2}+a
\end{array}\right)
\end{equation}
with non-zero coordinates (up to permutations)
\begin{equation}
\begin{aligned}
& x_{00}=1\\
& x_{11}=2[\mathrm{Re}(\gamma)-a],\, x_{12}=2\,\mathrm{Im}(\gamma),\, x_{13}=-2\sqrt{2}\,\mathrm{Re}(\beta)\\
& x_{22}=-2[\mathrm{Re}(\gamma)+a],\, x_{23}=2\sqrt{2}\,\mathrm{Im}(\beta),\, x_{33}=4a+1
\end{aligned}
\end{equation}
where $\beta,\gamma\in\mathbb{C}$ and $a\in\mathbb{R}$. Positivity of $\rho$ is however not yet guaranteed and imposes additional constraints on the values of $\beta,\gamma$ and $a$. If we ask that all principal minors of $\rho$ are nonnegative, which translates into the conditions
\begin{equation}
\label{positivityrho1}
\begin{aligned}
&a(1+2a)\leqslant -|\beta|^{2}\\
&2a\left(|\gamma|^{2}-|\beta|^{2}-\frac{1}{4}-a(1+a)\right)\geqslant |\beta|^{2}+2\,\mathrm{Re}[\gamma\beta^{*2}]
\end{aligned}\,,
\end{equation}
then the matrix $\rho$ is positive semi-definite and represents a possible spin--1 state anticoherent to order 1. Conversely, every spin--1 state anticoherent to order 1 has the form (\ref{rho1}) with $\beta,\gamma\in\mathbb{C}$ and $a\in\mathbb{R}$ verifying conditions (\ref{positivityrho1}). An expression for spin--$3/2$ anticoherent state to order 2 has been given in~[34].

\textbf{Anticoherent states to order 2.} 
From Theorem 3 and Eqs.~\eqref{defcoor} and \eqref{coorid}, spin-$j$ state anticoherent to order 2 is characterized by the fact that the matrix $A_{\mu\nu}\equiv \langle S_{\mu\nu 00\ldots 0}\rangle$ is given by
\begin{equation}
A=\left(\begin{array}{cccc}
1 & 0 & 0 & 0\\
0& \frac13 & 0&0\\
0&0&\frac13&0\\
0&0&0&\frac13
\end{array}\right).
\end{equation}